\documentclass[11pt]{article}
\usepackage[toc,page]{appendix}
\usepackage{authblk}

\usepackage{graphicx} 
\graphicspath{ {images/} }

\usepackage[scr=rsfs]{mathalpha}

\usepackage{subcaption}
\usepackage{booktabs}
\usepackage[ruled,vlined]{algorithm2e}
\usepackage{hyperref}
\usepackage{subfiles}
\usepackage{array}

\usepackage{amsfonts}
\usepackage{amsmath}
\usepackage{amsthm}
\usepackage{physics}

\newtheorem{theorem}{Theorem}[section]
\newtheorem{lemma}{Lemma}[section]

\usepackage{thm-restate}

\newcommand{\sS}{\mathcal{S}} 
\renewcommand{\H}{\mathcal{H}} 

\newcommand{\R}{\mathbb{R}} 

\newcommand{\ignore}[1]{}

\begin{document}

\author{Peyman Afshani}
\author{Pingan Cheng}
\affil{Department of Computer Science, Aarhus University, Denmark}
\affil{\{peyman, pingancheng\}@cs.au.dk}
\title{An Optimal Lower Bound for Simplex Range Reporting}
\maketitle

\medskip

\begin{abstract}

We give a simplified and improved lower bound for the simplex range reporting problem.
We show that given a set $P$ of $n$ points in $\R^d$, any data structure that uses
$S(n)$ space to answer such queries must have 
$Q(n)=\Omega((n^2/S(n))^{(d-1)/d}+k)$ query time, where $k$ is the output size.
For near-linear space data structures, i.e., $S(n)=O(n\log^{O(1)}n)$,
this improves the previous lower bounds by Chazelle and Rosenberg~\cite{cr1996s}
and Afshani~\cite{a2012i} but perhaps more importantly, it is the first
ever tight lower bound for any variant of simplex range searching for
$d\ge 3$ dimensions. 

We obtain our lower bound by making a simple connection to well-studied
problems in incident geometry which
allows us to use known constructions in the area. 
We observe that a small modification of a simple already existing construction can lead to our
lower bound. 
We believe that our proof is accessible to a much wider audience, 
at least compared to the previous intricate probabilistic proofs based on measure arguments by Chazelle and Rosenberg~\cite{cr1996s} and Afshani~\cite{a2012i}.

The lack of tight or almost-tight (up to polylogarithmic factor) lower bounds for near-linear space
data structures is a major bottleneck in making progress on problems such as proving
lower bounds for multilevel data structures.
It is our hope that this new line of attack based on incidence geometry can lead to further
progress in this area. 

\end{abstract}

\section{Introduction}

In the problem of simplex range reporting,
we are given a set $P$ of $n$ points in $\R^d$ as input
and we want to preprocess $P$ into a structure such that
given any query simplex $\gamma$,
we can report $P\cap\gamma$ efficiently. 
It is known that given $O(n)$ space,
the problem can be solved using $Q(n) = O(n^{1-1/d} + k)$ query time
where $k$ is the output size, i.e., $|P\cap\gamma|$~\cite{c2012o}.
However, current best lower bounds only match this upper bound
in the plane~\cite{cr1996s,a2012i} and the best known lower bound 
is off by a factor of $2^{O(\sqrt{\log n})}$ in higher dimensions~\cite{a2012i}.
Closing this gap has been a long-standing open problem
for this fundamental problem in computational geometry.

In this paper, we prove a tight query time lower bound for simplex range reporting 
in the pointer machine model in the case when the space usage is linear.
Our proof dramatically simplifies the previously known (suboptimal) 
proofs in~\cite{cr1996s} and~\cite{a2012i}.
We obtain the result by observing a connection to incidence geometry which 
allows us to use simple deterministic ``grid-based'' constructions and avoid
the intricate probabilistic construction and measure analysis used in the
previous proofs~\cite{a2012i,cr1996s}.

\subsection{Related Work.} 

Simplex range reporting is a classical and fundamental problem in computational geometry
and can be viewed as the most general case of range searching as far as linear constraints are concerned.
Indeed, by using multilevel data structures~\cite{agarwal2017simplex} and polyhedron triangulation,
any range intersection reporting problem with constant complexity linear inputs and queries reduces to simplex range reporting.
In general, there are many flavors of the problem.
Here, we focus on the reporting variant where given a query simplex, the goal is to output
the list of points inside the query, a.k.a. ``simplex range reporting''.
However, counting variants are also well-studied where the points have weights from a semi-group
and given the query, the goal is to output the sum of the weights of the points inside
the query, a.k.a. ``simplex range searching''. 

We now quickly review the history of the problem.
All the upcoming results apply to both variants of the problem. 
When discussing a  data structure, we use $S(n)$ to refer to the space complexity
and $Q(n)$ to refer to the query time (ignoring the time required to produce the output). 
Thus, with our notation, a data structure for simplex range reporting uses $S(n)$ space and 
it has the query time of $Q(n) + O(k)$. 

The first nontrivial result for the problem dates back to the early 1980s~\cite{willard1982polygon}.
After many early attempts~\cite{ew1986h,y1983a,yde1989p,a1984n,c1985p,yy1985a,hw1987e,w1988p,cw1989q},
significant progress was made after the discovery of 
fundamental tools such as the partition theorem~\cite{matousek1991efficient,m1993r,c2012o}
and cutting lemma~\cite{matousek1991cutting,chazelle1993cutting,deberg1995cuttings}. 
The first near-optimal solution of $O(n^{1+\varepsilon})$ space and $O(n^{1-1/d+\varepsilon}+k)$ query time\footnote{In this paper, 
$\varepsilon$ is an arbitrarily small positive constant.}
was found by Chazelle, Sharir, and Welzl~\cite{csw1992q} and it was simplified and slightly improved by Matou\v sek~\cite{m1993r}.
Finally, in 2012, Chan~\cite{c2012o} removed the $\varepsilon$ factors in the space and query time~\cite{csw1992q}.

It is clear from the above bounds that simplex range searching is a difficult problem since
using linear space, we can only improve the trivial query bound by an 
$n^{1/d}$ factor.
In 1989, Chazelle formally proved the difficulty of the problem by
showing a query time lower bound of $Q(n) = \Omega(n^{1-1/d}/\log n)$ 
for the general simplex range searching problem given linear space 
in the semigroup arithmetic model~\cite{c1989l}.
Unlike the upper bounds, this lower bound does not apply to the simplex range reporting problem.
Seven years later, Chazelle and Rosenberg~\cite{cr1996s} overcame this issue, and they 
showed that if the query time is $O(n^\delta + k)$, then the data structure must use
$\Omega(n^{d-d\delta-\varepsilon})$ space, where $k$ is the output size.
Note that the conjectured space-time trade-off for this problem is
$S(n) = O((n/Q(n))^d)$ and thus this lower bound is a factor
$n^{\varepsilon}$ factor away from this bound. 
It was observed by Afshani~\cite{a2012i} that another lower bound of
Chazelle and Liu~\cite{cl2001l} for the two-dimensional fractional cascading problem in fact
achieves the aforementioned conjecture space-time trade-off for simplex range reporting in the
plane ($d=2$). 
However, for $d\ge 3$, the only improvement is a lower bound by 
Afshani~\cite{a2012i} 
who showed a tighter query time lower bound of $\Omega(n^{1-1/d}/2^{O(\sqrt{\log n})})$~\cite{a2012i}
which narrows the gap from a polynomial ($n^{\varepsilon}$) factor to 
a sub-polynomial ($2^{O(\sqrt{\log n})}$) one. 
Completely eliminating this gap seems like a challenging problem since the techniques used
by the previous lower bounds inherently tie to a long-standing open problem known as the
Heilbronn's triangle problem~\cite{roth}.

The lack of tight lower bounds for the simplex range reporting problem is also a bottleneck
in trying to obtain lower bounds for some more complicated problems,  
for instance, for multilevel data structures (i.e., data structures that involve multiple levels
of simplex range searching data structures).

\subsection{Our Contribution.}
We simplify and improve the lower bound for simplex range reporting by Chazelle and Rosenberg~\cite{cr1996s} and Afshani~\cite{a2012i}.
Specifically, we show a lower bound of $Q(n)=\Omega((n^2/S(n))^{(d-1)/d}+k)$ for the problem.
When $S(n)=O(n)$, we get a clean lower bound of $Q(n)=\Omega(n^{(d-1)/d}+k)$,
which is the \textbf{first tight lower bound} for simplex range reporting for $d\ge 3$. 
By a known technique~\cite{a2012i},
our result also improves the lower bound for halfspace range reporting in $9$ and higher dimensions.
Along the way, we made the observation that the point-hyerplane incidence problem
is highly related to proving lower bounds for simplex range reporting.

\section{Preliminaries of the Pointer Machine Lower Bound Framework}

We will prove the lower bound for simplex range reporting 
in (an augmented version of) the pointer machine model.
In this model,
the data structure is modeled as a directed graph $M$.
In each cell of $M$, we store an element of the input set $\sS$
as well as two pointers to other cells.
To find the answer to a query $q$, i.e., a subset $\sS_q\subset \sS$,
the algorithm starts at a special ``root'' cell
and explores a connected subgraph such that all elements in $\sS_q$
can be found in some cell in the subgraph.
During the process, we only charge for pointer navigations.
Let $M_q$ be the smallest connected subgraph in which
every element of $\sS_q$ is stored at least once.
Clearly, $|M|$ is a lower bound for the space usage and $|M_q|$ is a lower bound for the query time. 
Note that this grants the algorithm unlimited computational power as well as
full information about the structure of $M$.

We use the following pointer machine lower bound framework tailored for geometric range reporting problems
by Chazlle~\cite{c1990lii} and Chazelle and Rosenberg~\cite{cr1996s}.

\begin{theorem}[Chazlle~\cite{c1990lii} and Chazelle and Rosenberg~\cite{cr1996s}]
\label{thm:pm-fw}
	Suppose there is a data structure of space $S(n)$
	which can answer range reporting queries in time $Q(n) + O(k)$
	where $n$ and $k$ are the input and output sizes respectively.
	Assume we can show the existence of a set $\sS$ of $n$ points
	such that there exist $m$ subsets $q_1,q_2,\cdots,q_m\subset\sS$,
	where $q_i,i=1,2,\cdots,m$, is the output of some query and they
	satisfy the following two conditions:
	(i) for all $i=1,2,\cdots,m$, $|q_i|\ge Q(n)$;
	and (ii) the size of the intersection of every $\beta\ge 2$ distinct
	subsets $q_{i_1}, q_{i_2}, \cdots,q_{i_{\beta}}$
	is upper bounded by some value $\alpha$, i.e., $|q_{i_1}\cap q_{i_2}\cap\cdots\cap q_{i_{\beta}}|\le \alpha$.
	Then $S(n)=\Omega(\frac{\sum_{i=1}^m|q_i|}{\beta 2^{O(\alpha)}})=\Omega(\frac{mQ(n)}{\beta 2^{O(\alpha)}})$.
\end{theorem}

\section{A Lower Bound for Simplex Range Reporting}

\subsection{Simplex Range Reporting Lower Bounds Through the Incidence Geometry Lens.}

Now we proceed to prove the lower bound.
Our first observation is that
to get a lower bound for simplex range reporting,
we only need to study a specific incidence geometry problem.
This is due to the fact that hyperplanes are degenerated simplicies, 
and so to show a lower bound for simplex range reporting using Theorem~\ref{thm:pm-fw},
it suffices to give a point-hyperplane configuration satisfying the two conditions in Theorem~\ref{thm:pm-fw}.
Stated in the language of incidence geometry, the first condition
requires each hyperplane to be incident to enough (at least $Q(n)$) points.
The second condition requires us to bound the size of $K_{\alpha,\beta}$
in the incidence graph.
To put it more formally,
Theorem~\ref{thm:pm-fw} implies the following lemma:

\begin{lemma}
\label{lem:inc-lb}
	If there exist a set $P$ of $n>0$ points
	and a set $H$ of $m>0$ hyperplanes each incident to at least $t\ge Q(n)$ 
	points (called $t$-rich hyperplanes) in $\R^d$
	with no complete bipartite subgraph $K_{\alpha,\beta}$
	in the incidence graph $P\times H$,
	then the simplex range reporting problem
	has a lower bound of $S(n)=\Omega(\frac{mt}{\beta 2^{O(\alpha)}})=\Omega(\frac{mQ(n)}{\beta 2^{O(\alpha)}})$.
\end{lemma}

It turns out that the relationship between the number of point-hyperplane incidences
and $K_{\alpha,\beta}$ is a well-studied problem in the incidence geometry community~\cite{bk2003o, as2007l, s2016l, bcv2019c}.
However, this is not directly relevant to us as we require each hyperplane to be ``rich''.
The closest result of the problem we can find is the very recent work by Pat\'akov\'a and Sharir~\cite{ps2022c}.
They showed the existence of $n$ points and $\Theta(n^d/t^{d+1})$ $t$-rich hyperplanes with no $K_{\alpha,\beta}$
in the incidence graph for $\beta = 2$ and $\alpha = O(t^{(d-2)/(d-1)})$.
They also showed a matching lower bound for the size of $\alpha$ given $\beta=2$.

Unfortunately, their result does not give us a useful lower bound.\footnote{
In fact, by plugging the parameters in~\cite{ps2022c} in Lemma~\ref{lem:inc-lb},
we can only get a lower bound of $Q(n)=\Omega((\log\frac{n^d}{S(n)})^{\frac{d-1}{d-2}})$.
}
The main reason for this is that the lower bound in Lemma~\ref{lem:inc-lb}
has a $2^{O(\alpha)}$ factor in the denominator and so
to show a nontrivial lower bound,
$\alpha$ has to be sub-logarithmic.
In our proof, we will still use the construction in~\cite{ps2022c},
but we prove an upper bound for $\beta$ by fixing $\alpha=2$.
Note that this is the opposite to the case considered in~\cite{ps2022c}.

\subsection{A Simple Point-Hyperplane Incidence Geometry Lemma.}

Here, we prove the following lemma:

\begin{lemma}
\label{lem:inc-geo}
	 There exists a configuration of $n$ points and $m=\Theta(n^{d}/t^{d+1})$
	 $t$-rich hyperplanes with no $K_{2,\beta}$
	 in the incidence graph where
	 $\beta=\Theta(n^{d-2}/t^{d(d-2)/(d-1)})$
	 for any positive integer $t \le cn^{1-1/d}$
	 for some small enough positive constant $c$.
\end{lemma}

We consider the same construction in~\cite{ps2022c}.
For the completeness and readability, we present the construction and reprove some basic facts we will use.
W.l.o.g., we assume that $t^{1/(d-1)}$ and $n/t$ are integers; otherwise we can increase $t$ and decrease $n$ 
slightly to ensure the assumption.
(It can be easily shown that $t,n$ will remain asymptotically the same after the process.)
Let $G$ be an integer grid in $\R^d$ of size $t^{1/(d-1)} \times t^{1/(d-1)} \times \cdots \times t^{1/(d-1)} \times n/t$.
Clearly, $G$ has $n$ grid points.
We construct hyperplanes of form
\[
	X_d=b + \sum_{i=1}^{d-1}a_iX_i,
\]
where $a_i\in\{1,\cdots,A\}$ and $b\in\{1,\cdots,B\}$ for $A=\lfloor\frac{n}{dt^{d/(d-1)}}\rfloor$ and $B=\lfloor\frac{n}{dt}\rfloor$.
Since $t\le cn^{1-1/d}$ for a small enough positive constant $c$, $A,B\ge 1$ and so our construction is valid.
We create all the possible distinct hyperplanes by picking $a_i$'s and $b$ as above. 
Let $\H$ be the set of all the hyperplanes we generated this way. 
As we have $A$ choices for each $a_i$ and $B$ choices for $b$, 
the total number of hyperplanes we generated is $m=|\H| = A^{d-1}B=\Theta(n^d/t^{d+1})$.

Now consider a hyperplane $h_j \in \H$ and its intersection with $G$.
Observe that all the coefficients of $h_j$ are positive integers.
This means that plugging in an integer value $x_i$ 
for $X_i$ for $i=1, \cdots, d-1$ will yield the integer value $x_d = b + \sum_{i=1}^{d-1}a_i x_i$ thus a point
$(x_1, \cdots, x_d)$ with integer coordinates that lies on $h_j$.
The value $x_d$ is maximized when $b$ is set to $B$ and all $a_i$'s are set to $A$.
Furthermore, the largest value of the first $d-1$ dimensions of $G$ is $t^{1/(d-1)}$.
Since
\[
	B+(d-1)At^{1/(d-1)} \le n/t,
\]
each hyperplane in $\H$ intersects exactly $(t^{1/(d-1)})^{d-1}=t$ grid points.

Finally, we bound $\beta$ given $\alpha=2$.
We use the following simple lemma.
This is the only new property we show in this construction
and it has a very simple proof.
\begin{lemma}\label{lem:int}
	Any subset $\H' \subset \H$ of size $|\H'|\ge A^{d-2} + 1$
	contains at most one point in common.
\end{lemma}
\begin{proof}
  We do proof by contradiction. 
  Assume hyperplanes in $\H'$ have two distinct points $g_1$ and $g_2$ in common, 
  then there must be at least one coordinate on which they differ.
  Note that the $d$-th coordinate cannot be the only difference between $g_1$ and $g_2$
  because hyperplanes in $\H$ are not parallel to the $d$-th axis.
  W.l.o.g., we can assume that $g_1$ and $g_2$ differ in their $(d-1)$-th coordinate.
  By the pigeonhole principle, there will be two hyperplanes $h_1, h_2 \in \H'$ that have identical
  first $d-2$ coefficients. 
  Assume $h_1$ is defined by coefficients $a_1, \cdots, a_{d-2}, a_{1,d-1}, b_1$
  and $h_2$ is defined by coefficients $a_1, \cdots, a_{d-2}, a_{2,d-1}, b_2$.
  We can view $h_1$ and $h_2$ as linear functions, $f_1$ and $f_2$, from $\R^{d-1}$ to $\R$.
  Let $X^{(d-1)} = (X_1, \cdots, X_{d-1})$.
  We thus write 
  \[
    f_1(X^{(d-1)}) = b_1 + a_{1,d-1}X_{d-1}+ \sum_{i=1}^{d-2} a_i X_i 
  \]
  and 
  \[
    f_2(X^{(d-1)}) = b_2 + a_{2,d-1}X_{d-1}+ \sum_{i=1}^{d-2} a_i X_i.
  \]
  Consider the function
  \[
    D(X^{(d-1)}) = f_1(X^{(d-1)}) -f_2(X^{(d-1)}) = b_1-b_2 + (a_{1,d-1}-a_{2,d-1})X_{d-1}.
  \]
  Let $g'_1$ and $g'_2$ be the projection of $g_1$ and $g_2$ onto the first $d-1$ dimensions.
  Since $h_1$ and $h_2$ pass through points $g_1$ and $g_2$ we have
  \[
    D(g'_1) = D(g'_2) = 0.
  \]
  However, the function $D(\cdot)$ is essentially a univariate linear function (i.e., a line
  in the coordinate system defined by the $(d-1)$-th and $d$-th axes). 
  Furthermore, since $g_1$ and $g_2$ have distinct $(d-1)$-th coordinates, it follows that this function
  is zero on two distinct points. 
  This implies that the function $D(\cdot)$ must be identical to the zero function which implies
  $h_1 = h_2$, a contradiction. 
  Thus, the lemma follows. 
\end{proof}
According to Lemma~\ref{lem:int}, there is no $K_{2,\beta}$ in the incidence graph of our construction
for 
\[
	\beta = A^{d-2} + 1 = \Theta(n^{d-2}/t^{d(d-2)/(d-1)}).
\]

This completes the proof of Lemma~\ref{lem:inc-geo}.

\subsection{Combining Them Together.}

Now we are ready to show a lower bound for simplex range reporting.
Suppose $\lceil Q(n) \rceil < cn^{1-1/d}$, where $c$ is the constant in Lemma~\ref{lem:inc-geo}, 
then we can set $t=\lceil Q(n) \rceil$ and Lemma~\ref{lem:inc-geo} applies.
By Lemma~\ref{lem:inc-lb}, we obtain a lower bound of
\[
	S(n) = \Omega\left(
			\frac{
				\Theta\left(\frac{n^d}{Q(n)^{d+1}}\right) \cdot Q(n)
			}{
				\Theta\left( \frac{ n^{d-2} }{ Q(n)^{d(d-2)/(d-1)}} \right) \cdot 2^{O(2)}
			}
		\right)
		= \Omega\left(
			\frac{n^2}{Q(n)^{\frac{d}{d-1}}}
		\right)
		\implies
		Q(n) = \Omega\left( \left( \frac{n^2}{S(n)} \right)^{\frac{d-1}{d}} \right).
\]
On the other hand, if $\lceil Q(n) \rceil \ge cn^{1-1/d}$, then there is nothing to prove since this is already a lower bound.
To sum up, we have proved the following theorem:
\begin{theorem}
	The simplex range reporting problem has a lower bound of $Q(n)=\Omega((n^2/S(n))^{(d-1)/d})$.
\end{theorem}

\section{Open Problems}

There are three main open problems.
The first and the major open problem is to show a tight lower bound
for super-linear space data structures for simplex range reporting.
Our current construction is only optimal
when the space usage is restricted to linear.
Although it is one of the most important cases for the problem,
it would be desirable to obtain a tight space-time tradeoff.
The main challenge here is to generate more $t$-rich hyperplanes
without increasing $\beta$ too much while restricting $\alpha$ to be small, say a constant.

Second, it is open if we can achieve tight lower bounds
for other models of computation.
For example, can we get a tight query time lower bound for
the general simplex range searching problem 
in the semigroup arithmetic model given linear space?
In this model, it is also possible to formulate
a lower bound framework based on the point-hyperplane incidence property.
But in this case, we need to bound $\alpha$ such that its value
decreases proportional to $\beta$.
See~\cite{c1990lii,c2001t} for the classical lower bound framework in this model.
Unfortunately, our construction does not have this property.

Finally, it is interesting to see if such improvement can be made
in related problems like multilevel data structures
as well as the dual stabbing problems.

\bibliography{reference}{}
\bibliographystyle{abbrv}

\end{document}